\documentclass[letterpaper, 10 pt, conference]{ieeeconf}  

\IEEEoverridecommandlockouts                              

\usepackage{graphics} 
\usepackage{epsfig} 
\usepackage{times} 
\usepackage{amsmath} 
\usepackage{amssymb}  
\usepackage{paralist}
\usepackage{algorithmic,algorithm}
\usepackage{subfig}
\usepackage{xcolor}
\usepackage[colorlinks=true,linkcolor=black]{hyperref}
\usepackage{comment}
\newtheorem{theorem}{Theorem}

\newtheorem{lemma}{Lemma}

\newtheorem{rem}{Remark}

\title{
Hybrid Feedback Control for Global and Optimal Safe Navigation
}

\author{Ishak Cheniouni, Soulaimane Berkane and Abdelhamid Tayebi 
	\thanks{This work was supported in part by the National Sciences and Engineering Research Council of Canada (NSERC), under the grants NSERC-DG RGPIN 2020-06270 and NSERC-DG RGPIN-2020-04759, and Fonds de recherche du Qu\'ebec (FRQ).} 
	\thanks{ The authors are with the Department of Electrical Engineering,  Lakehead University, Thunder Bay, ON P7B 5E1, Canada (e-mail:  cheniounii, atayebi@lakeheadu.ca)}
	\thanks{ Abdelhamid Tayebi is also with the Department of Electrical and Computer Engineering, Western University, London, ON N6A 3K7, Canada}%
	\thanks{S. Berkane is also with the Department Computer Science and Engineering,  University of Quebec in Outaouais, QC J8X 3X7, Canada (e-mail: soulaimane.berkane@uqo.ca).}%
}

\begin{document}

\maketitle
\thispagestyle{empty}
\pagestyle{empty}

\begin{abstract}
We propose a hybrid feedback control strategy that safely steers a point-mass robot to a target location {\it optimally from all initial conditions} in the $n$-dimensional Euclidean space with a single spherical obstacle. The robot moves straight to the target when it has a clear line-of-sight to the target location. Otherwise, it engages in an optimal obstacle avoidance maneuver via the shortest path inside the cone enclosing the obstacle and having the robot's position as a vertex. The switching strategy that avoids the undesired equilibria, leading to global asymptotic stability (GAS) of the target location, relies on using two appropriately designed virtual destinations, ensuring control continuity and shortest path generation. Simulation results illustrating the effectiveness of the proposed approach are presented.
\end{abstract}

\section{INTRODUCTION}
Autonomous robot navigation with obstacle avoidance has come a long way since the artificial potential field approach of Khatib \cite{khatib}. This pioneering approach, which considers the destination as an attractive force and the obstacles as repulsive forces, although simple and elegant, does not guarantee that the target (minimum of the potential) will be reached, as local minima are often generated. This problem was solved with the concept of the navigation function introduced in \cite{k_R_90}, where almost global asymptotic stability (AGAS) of the target location is achieved in sphere worlds. Using diffeomorphisms, the same results are obtained in more complex worlds  \cite{R_k_92}. A tuning-free navigation function has been constructed in point worlds for navigation in sphere and star worlds with AGAS guarantees \cite{LoizouNT4, LoizouNT3}.  In \cite{SnsNF6}, a sufficient condition is provided for valid navigation functions in environments with convex obstacles that are relatively flat with respect to their distance to the target. Under the same condition, the authors of \cite{Arslan2019} proposed an almost global reactive sensor-based approach that uses the Voronoi diagram to decompose the workspace where the robot cell is constructed by the intersection of the hyperplanes separating the robot from adjacent obstacles. The projection of the target on this convex cell plays the role of an intermediary local destination for the robot. This work has been extended to environments with star-shaped and polygonal obstacles with possible overlap in \cite{vasilopoulos1,Vasilopoulos2}. Path optimality has long been one of the main objectives of path planning algorithms, while feedback-based approaches such as those mentioned above generally do not prioritize the shortest path. In \cite{ACC23,Ishak2023}, we proposed a continuous feedback control strategy in sphere worlds generating quasi-optimal paths. Unfortunately, sets of non-zero measure, from where the undesired equilibria may be reached, may exist unless some restrictions on the obstacles configuration in the workspace are enforced. A sensor-based version of this approach, with AGAS guarantees of the target location, has also been proposed in \cite{Ishak2023} to deal with convex obstacles satisfying a flatness condition similar to the one in \cite{SnsNF6,Arslan2019}. A global result is out of reach for the above-cited works, involving smooth controllers, due to topological obstruction pointed out in \cite{k_R_90}. As an alternative, hybrid feedback controllers have been proposed in the literature to achieve GAS results. The work in \cite{HybBerkaneECC2019} achieves GAS of the target location in Euclidean spaces with a single spherical obstacle. In \cite{SoulaimaneHybTr}, an extension was proposed for multiple ellipsoidal obstacles in Euclidean spaces. Recently, a hybrid feedback control for safe and global navigation in two-dimensional environments with arbitrary convex obstacles was proposed in \cite{Mayur2022}. A similar control scheme was proposed more recently in \cite{Mayur2023} for environments with non-convex obstacles. These hybrid feedback strategies endowed with GAS guarantees do not take path length optimality into account in their design. This lack of optimality motivates us to leverage the results achieved in \cite{Ishak2023} together with the hybrid feedback techniques \cite{HybBerkaneECC2019} to design a feedback control scheme that enjoys GAS properties, while generating optimal obstacle avoidance trajectories, in terms of the Euclidean distance, in an $n$-dimensional Euclidean space with a single spherical obstacle. 

\section{Notations and Preliminaries}
Throughout the paper, $\mathbb{N}$ and $\mathbb{R}$ denote the set of natural numbers and real numbers, respectively. The Euclidean space and the unit $n$-sphere are denoted by $\mathbb{R}^n$ and $\mathbb{S}^n$, respectively. The Euclidean norm of $x\in\mathbb{R}^n$ is defined as $\|x\|:=\sqrt{x^\top x}$ and the angle between two non-zero vectors $x,y\in\mathbb{R}^n$ is given by $\angle (x,y):=\cos^{-1}(x^\top y/\|x\|\|y\|)$ . 
We define the ball centered at $x\in\mathbb{R}^n$ and of radius $r>0$ by the set $\mathcal{B}(x,r):=\left\{q\in\mathbb{R}^n|\;\|q-x\| \leq r\right\}$. The interior, the boundary, and the closure of a set $\mathcal{A}\subset\mathbb{R}^n$ are denoted by $\mathring{\mathcal{A}}$, $\partial\mathcal{A}$, and $\overline{\mathcal{A}}$, respectively. The relative complement of a set $\mathcal{B}\subset\mathbb{R}^n$ with respect to a set $\mathcal{A}$ is denoted by $\mathcal{B}^c_\mathcal{A}$. 
The line passing by two points $x,y\in\mathbb{R}^n$ is defined as $\mathcal{L}(x, y):=\left\{q\in\mathbb{R}^n|q=x+\delta(y-x),\;\delta\in\mathbb{R}\right\}$. The elementary reflector, parallel projection, and orthogonal projection are, respectively, defined as follows \cite{Matrix_Analysis_and_Applied_Linear_Algebra}:
\begin{align*}
    \rho(v):=I_n-2vv^\top,\,\pi^{\parallel}(v):=vv^\top,\,\pi^{\bot}(v)&:=I_n-vv^\top,
\end{align*}
where $I_n\in\mathbb{R}^{n\times n}$ is the identity matrix and $v\in\mathbb{S}^{n-1}$. Therefore, for any vector $x$, the vectors $\rho(v)x$, $\pi^{\parallel}(v) x$ and $\pi^{\bot}(v) x$ correspond, respectively, to the reflection of $x$ about the hyperplane orthogonal to $v$, the projection of $x$ onto the line generated by $v$, and the projection of $x$ onto the hyperplane orthogonal to $v$.
Let us define the set $\mathcal{P}_{\Delta}(x,v)=\left\{q\in\mathbb{R}^n|v^\top(q-x)~\Delta~0\right\}$, with $\Delta \in\{=,>,\geq,<,\leq\}$.
The hyperplane passing through $x\in\mathbb{R}^n$ and orthogonal to $v\in\mathbb{R}^n\setminus\{0\}$ is denoted by $\mathcal{P}_{=}(x,v)$. The closed negative half-space (resp. open negative half-space) is denoted by $\mathcal{P}_{\leq}(x,v)$ (resp. $\mathcal{P}_{<}(x,v)$) and the closed positive half-space (resp. open positive half-space) is denoted by $\mathcal{P}_{\geq}(x,v)$ (resp. $\mathcal{P}_{>}(x,v)$). A conic subset of $\mathcal{A}\subseteq\mathbb{R}^n$, with vertex $x\in\mathbb{R}^n$, axis $a\in\mathbb{R}^n$, and aperture $2\varphi$ is defined as follows \cite{HybBerkaneECC2019}:
\begin{align}
    \mathcal{C}^{\Delta}_{\mathcal{A}}(x,a,\varphi):=\left\{q\in\mathcal{A}|\|a\|\|q-x\|\cos(\varphi)\Delta a^\top(q-x)\right\},
\end{align}
where $\varphi\in(0,\frac{\pi}{2}]$ and $\Delta\in\left\{\leq,<,=,>,\geq\right\}$,  with $``="$, representing the surface of the cone, $``\leq"$ (resp. $``<"$) representing the interior of the cone including its boundary (resp. excluding its boundary), and $``\geq"$ (resp. $``>"$) representing the exterior of the cone including its boundary (resp. excluding its boundary). We state a property of cones sharing the same vertex as follows \cite[Lemma~1]{HybBerkaneECC2019}:
\begin{lemma}\label{lem1}
Let $c,a_{-1},a_{1}\in\mathbb{R}^n$ such that $\angle(a_{-1},a_{1})=\psi$ where $\psi\in(0,\pi]$. Let $\varphi_{-1},\varphi_{1}\in[0,\pi]$ such that $\varphi_{-1}+\varphi_{1}<\psi<\pi-(\varphi_{-1}+\varphi_{1})$. Then 
\begin{align}
    \mathcal{C}_{\mathbb{R}^n}^{\leq}(c,a_{-1},\varphi_{-1})\cap\mathcal{C}_{\mathbb{R}^n}^{\leq}(c,a_{1},\varphi_{1})=\{c\}.
\end{align}
\end{lemma}
A hybrid dynamical system \cite{Sanfelice} is represented by
\begin{align} {\begin{cases}\dot{X}\in \mathrm{F}(X), &X\in \mathcal {F}\\ X^+\in \mathrm{J}(X), & X\in \mathcal {J} \end{cases}} \label{hyb} \end{align}
 where $X \in \mathbb{R}^n$ is the state, the (set-valued) flow map $\mathrm{F} : \mathbb{R}^n \rightrightarrows \mathbb{R}^n$
and jump map $\mathrm{J} : \mathbb{R}^n \rightrightarrows \mathbb{R}^n$ govern continuous and discrete evolution,
which can occur, respectively, in the flow set $\mathcal{F} \subset \mathbb{R}^n$ and the jump set
$\mathcal{J} \subset \mathbb{R}^n$. The notions of solution $\phi$ to a hybrid system, its hybrid time domain $\text{dom}\,\phi$, maximal and complete solution are, respectively, as
in \cite[Def. 2.6, Def. 2.3, Def. 2.7, p. 30]{Sanfelice} .
\section{Problem Formulation}
We consider a point mass vehicle $x\in\mathbb{R}^n$ and a spherical obstacle $\mathcal{O}:=\mathcal{B}(c,r)$. The free space is, therefore, given by
\begin{align}\label{1}
\mathcal{X}:=\mathbb{R}^n\setminus\mathring{\mathcal{O}}.
\end{align}
We consider the following first-order vehicle dynamics
\begin{align}\label{12}
    \Dot{x}=u,
\end{align}
where $u$ is the control input. The objective is to determine a state feedback controller $u(x)$ that safely and optimally steers the vehicle from any initial position $x(0)\in\mathcal{X}$ to any given desired destination $x_d\in\mathring{\mathcal{X}}$. In particular, the closed-loop system
\begin{equation}\label{eq:closed-loop-system}
    \dot x=u(x),\quad x(0)\in\mathcal{X}
\end{equation}
must ensure forward invariance of the set $\mathcal{X}$, global asymptotic stability of the equilibrium $x=x_d$, and the trajectory $x(t)$ must generate the shortest path from $x(0)$ to $x_d$.
\section{Sets Definition}\label{section:sets}
In this section, we define the subsets of the free space that are necessary for the control design that we propose in Section \ref{section:control-design}. These subsets are illustrated in Fig. \ref{fig:set_def} and given as follows:
\begin{itemize}
    \item The \textit{shadow region} is the area where the vehicle does not have a clear line-of-sight to the target. This is defined as follows:
    \begin{multline}\label{18}
        \mathcal{S}(x_d):=\bigl\{q\in\mathcal{C}^{\leq}_{\mathcal{X}}(x_d,c-x_d,\theta(x_d))|\\(c-q)^\top(x_d-q)\geq0\bigr\},
    \end{multline}
    where the function $\theta(q)\colon\mathcal{X}\to (0,\frac{\pi}{2}]$, $q\mapsto\theta(q):=\arcsin(r/\|q-c\|)$ assigns to each position $q$ of the free space, the half aperture of the cone enclosing obstacle $\mathcal{O}$.
    \item The \textit{exit set} is the lateral surface of the shadow region  and is defined as follows:
    \begin{multline}\label{12}
        \mathcal{E}(x_d):=\bigl\{q\in\mathcal{C}^{=}_{\mathcal{X}}(x_d,c-x_d,\theta(x_d))|\\(c-q)^\top(x_d-q)\geq0\bigr\},
    \end{multline}
    \item The \textit{visible set} which is defined by $\mathcal{V}(x_d):=\mathcal{S}_{\mathcal{X}}^c(x_d)$.
\end{itemize}
\begin{figure}[!h]
    \centering
    \includegraphics[scale=0.35]{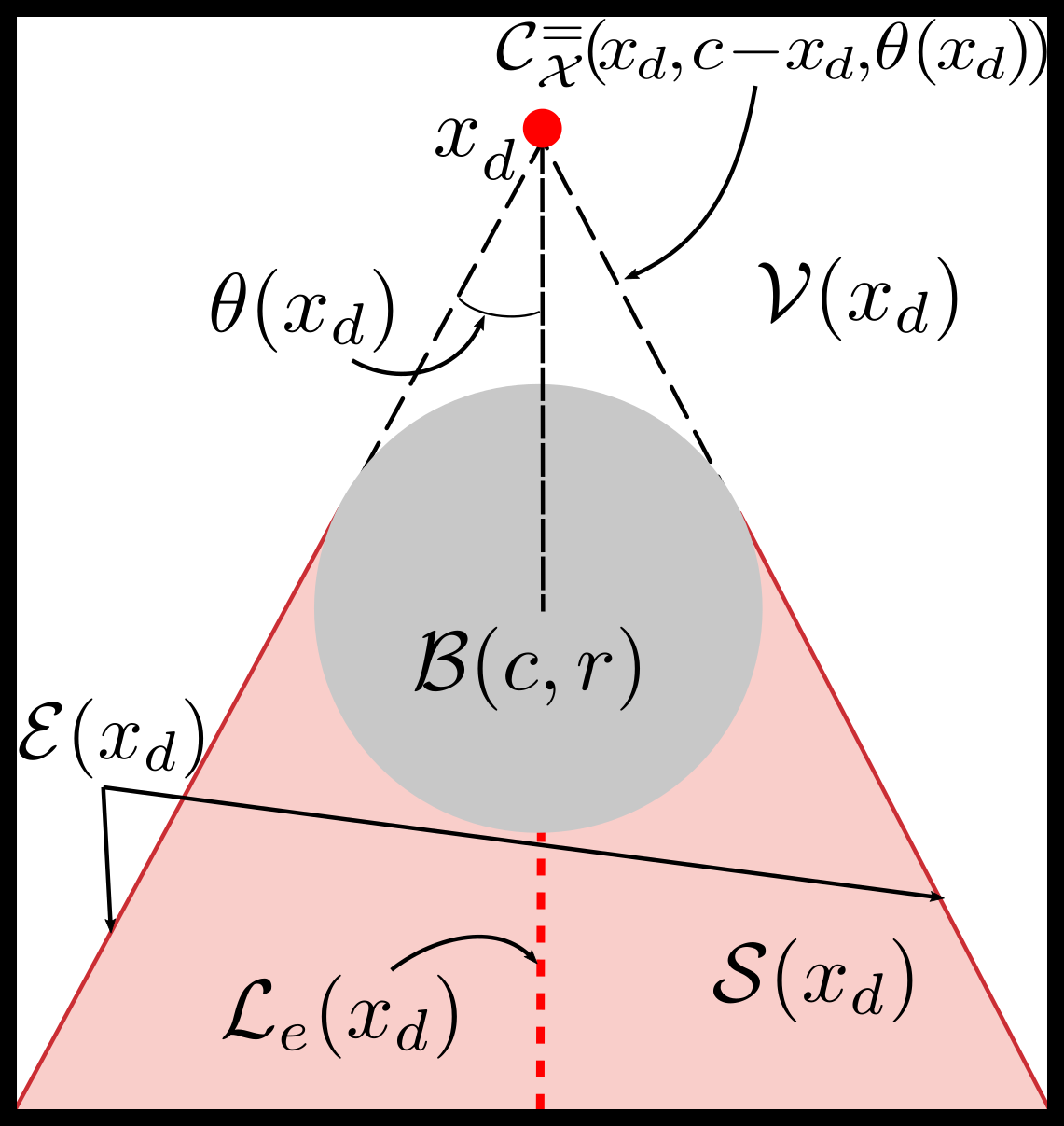}
    \caption{Illustration of the different subsets of the free space defined in Section \ref{section:sets}.}
    \label{fig:set_def}
\end{figure}
\section{Control Design}\label{section:control-design}
The objective is to reach the target safely by the shortest path from everywhere in the free space $\mathcal{X}$. Such an approach was proposed in \cite{Ishak2023}, where almost global asymptotic stability of a target location is achieved in a workspace with a single spherical obstacle through the following control strategy:
\begin{align}\label{18}
    u(x)=\begin{cases}
        u_d(x),&x\in\mathcal{V}(x_d)\\
        u_d(x)-\tau(x)\frac{c-x}{\|c-x\|},&x\in\mathcal{S}(x_d)
    \end{cases}
\end{align}
where $\tau(x)=\|u_d(x)\|\sin(\theta(x)-\beta(x))\sin^{-1}(\theta(x))$ with $\theta(x)$ is the half aperture of the cone enclosing the obstacle, $\beta(x)=\angle(c-x,u_d(x))$, $u_d(x)=\gamma(x_d-x)$ is the nominal controller, and $\gamma>0$. The set of equilibria of the closed-loop system \eqref{eq:closed-loop-system}-\eqref{18} is given by $\zeta=\{x_d\}\cup\mathcal{L}_e(x_d)$ where $\mathcal{L}_e(x_d):=\mathcal{L}(x_d,c)\cap\mathcal{S}(x_d)$ is the set of undesired equilibria. The control generates shortest paths and makes the equilibrium point $x_d$ AGAS. Although unstable, the presence of undesired equilibria can lead to performance degradation and robustness issues \cite{berkane2017design}. Inspired by \cite{HybBerkaneECC2019}, we propose a hybrid version of the control \eqref{18}, endowed with GAS, which will be presented next.
\subsection{Hybrid Control}
Consider obstacle $\mathcal{O}$ and the hybrid dynamical system \eqref{hyb}. We introduce a mode indicator $m\in\{-1,0,1\}$ where $m=0$ refers to the {\it straight} mode (when $x\in\mathcal{V}(x_d)$) and $m\in\{-1,1\}$ refers to the {\it projection} mode (when $x\in\mathcal{S}(x_d)$). We define two virtual destinations $x_d^m\in\mathcal{C}^{=}(x_d,c-x_d,\theta(x_d))\setminus\mathcal{E}(x_d)$ on the hat of the cone enclosing obstacle $\mathcal{O}$ where $\|x_d-x_d^m\|=e$ and $x_d^{-m}=x_d-\rho(\frac{c-x_d}{\|c-x_d\|})(x_d^m-x_d)$ for $m\in\{-1,1\}$ with a design parameter $e>0$. For each destination, we proceed with the  construction of the control law similar to \eqref{18} for each shadow region (see Fig. \ref{fig:hybShadow}). We propose the following feedback control law depending on the current position $x\in\mathbb{R}^n$ and the mode $m\in\{-1,0,1\}$:
\begin{align}\label{hyb_ctrl}
    u(x,m)&=|m|\mu(x,m)u_m(x)+(1-|m|)u_d^{m}(x),
\end{align}
where $u_m(x)=u_d^m(x)-\tau_m(x)\frac{c-x}{\|c-x\|}$, $u_d^m(x)=\gamma(x_d^m-x)$, $x_d^0=x_d$, $\tau_m(x)=\|u_d^m(x)\|\sin(\theta(x)-\beta_m(x))\sin^{-1}(\theta(x))$, $\beta_m(x)=\angle(c-x,u_d^m(x))$, and $\mu(x,m)=\left(1+\frac{e}{\|x-x_d^m\|}\frac{\beta_m(x)}{\theta(x)}\right)$. In the {\it straight} mode ($m=0$), the control law $u(x,0)=\gamma(x_d-x)$ steers the robot straight to the target. In the {\it projection} mode ($m=\pm1$), the control input $u(x,m)=\mu(x,m)u_m(x)$ selects the shortest path within the cone enclosing obstacle $\mathcal{O}$. The idea behind the scaling function $\mu(x,m)$ and the placement of the virtual destination on the hat of the cone enclosing obstacle $\mathcal{O}$, will be discussed later. Let us look for the equilibria of the closed-loop system in the {\it projection} mode. By setting $u(x,m)=0$ for $m\in\{-1,1\}$ and $x\in\mathcal{S}(x_d^m)$, one obtains $\mu(x,m)=0$ or $u_m(x)=0$, but the first term is always positive $\mu(x,m)>1$. Therefore, $u(x,m)=0$, for $m\in\{-1,1\}$ and $x\in\mathcal{S}(x_d^m)$, if and only if $u_m(x)=0$ which corresponds to the undesired equilibria of the closed-loop system $\eqref{18}$ for the virtual destinations. One can conclude that the equilibria of the {\it projection} mode are $\mathcal{L}_e(x_d^m)$, $m\in\{-1,1\}$. 
\begin{figure}[!h]
\centering
\includegraphics[scale=0.3]{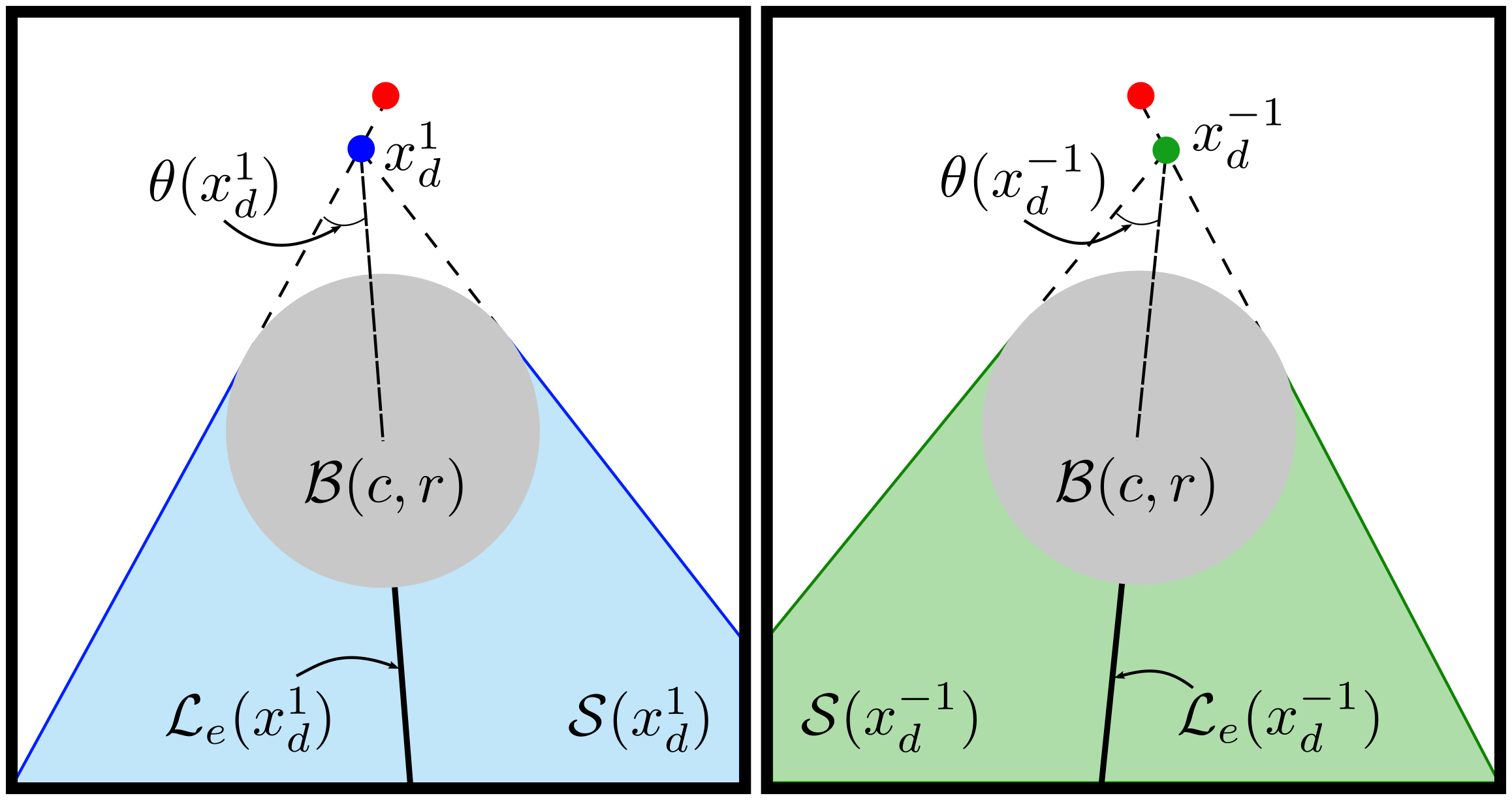}
\caption{Shadow regions of the virtual destinations.}
\label{fig:hybShadow}
\end{figure}
We define the flow and jump sets depicted in Fig. \ref{fig:flow_jump_sets} as follows:
\begin{subequations}\label{flow_jump_sets}
    \begin{align}
    \mathcal{F}_0&:=\overline{\mathcal{V}}(x_d),\;\mathcal{J}_0:=\mathcal{S}(x_d),\label{set1}\\
   \mathcal{F}_{1}&:=\mathcal{S}(x_d^{1})\setminus\mathcal{C}_{\mathcal{X}}^{<}(c,v_1,\varphi_1),\mathcal{J}_{1}:=\overline{\mathcal{F}_{1_{\mathcal{X}}}^{c}},\label{set2}\\
    \mathcal{F}_{-1}&:=\mathcal{S}(x_d^{-1})\setminus\mathcal{C}_{\mathcal{X}}^{<}(c,v_{-1},\varphi_{-1}),\mathcal{J}_{-1}:=\overline{\mathcal{F}_{-1_{\mathcal{X}}}^{c}},\label{set3}\\
    \mathcal{F}&:=\bigcup\limits_{i=-1}^{1}\left(\mathcal{F}_i\times\{m\}\right),\;\mathcal{J}:=\bigcup\limits_{i=-1}^{1}\left(\mathcal{J}_i\times\{m\}\right)\label{set4}
\end{align}
\end{subequations}
where $v_m:=c-x_d^m$ for $m\in\{-1,1\}$. For the well-posedness of the closed-loop hybrid system we design the angles $\varphi_1,\varphi_{-1}$ in \eqref{flow_jump_sets} as per Lemma \ref{lem1} as follows:
\begin{equation}\label{phi}
    \varphi_1=\varphi_{-1}=\varphi<\min\left\{\frac{\angle{(v_1,v_{-1})}}{2},\frac{\pi-\angle{(v_1,v_{-1})}}{2}\right\}.
\end{equation}
\begin{figure}[!h]
\centering
\includegraphics[scale=0.3]{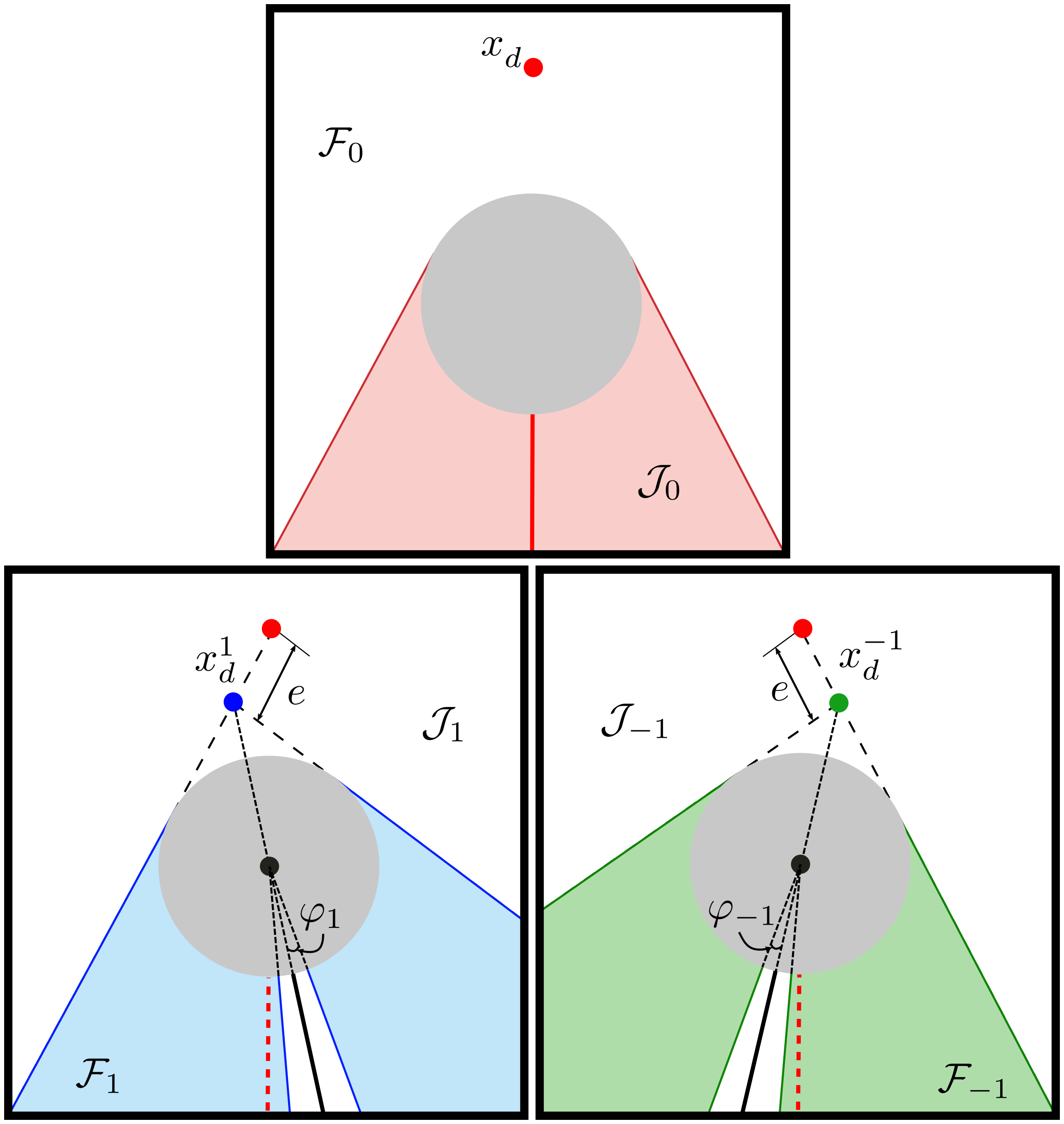}
\caption{The flow and jump sets for the proposed dynamical hybrid system.}
\label{fig:flow_jump_sets}
\end{figure}
The mode variable $m$ is selected according to the following scheme:
\begin{align}\label{jump_dyn}
    &\dot{m}=0,&&(x,m)\in\mathcal{F},\\
    &m^{+}\in M(x,m),&&(x,m)\in\mathcal{J},
\end{align}
where the jump map $M(\cdot)$ is defined as:
\begin{align}\label{jump_map}
    M(x,m):=\begin{cases}
    0 &x\in \mathcal{J}_m,\;m\in\{-1,1\},\\
    B(x)&x\in \mathcal{J}_0,\; m=0,
      \end{cases}
\end{align}
and $B(\cdot)$ is defined as:
\begin{align}\label{jump_map_B}
    B(x):=\begin{cases}
    1&x\in \mathcal{C}_{\mathbb{R}^n}^{\geq}(c,v_1,\varphi_1)\setminus\mathcal{C}_{\mathbb{R}^n}^{\geq}(c,v_{-1},\varphi_{-1}),\\
    -1 &x\in \mathcal{C}_{\mathbb{R}^n}^{\geq}(c,v_{-1},\varphi_{-1})\setminus\mathcal{C}_{\mathbb{R}^n}^{\geq}(c,v_1,\varphi_1),\\
    \{-1,1\}&x\in \mathcal{C}_{\mathbb{R}^n}^{\geq}(c,v_1,\varphi_1)\cap\mathcal{C}_{\mathbb{R}^n}^{\geq}(c,v_{-1},\varphi_{-1}).
      \end{cases}
\end{align}
The hybrid closed-loop system can be written as follows:
\begin{subequations}\label{Cl_ctrl}
    \begin{align}
    &\begin{cases}
    \Dot{x}=u(x,m)\\
    \Dot{m}=0
    \end{cases}
    &&(x,m)\in\mathcal{F},\label{CL_hyb1}\\
   &\begin{cases}
    x^{+}=x\\
    m^{+}\in M(x,m)
    \end{cases}
    &&(x,m)\in\mathcal{J},\label{CL_hyb2}
\end{align}
\end{subequations}
and its representation with the hybrid data is given by $\mathcal{H}:=(\mathcal{F},\mathrm{F},\mathcal{J},\mathrm{J})$ where $\mathrm{F}(x,m)=\left(u(x,m),0)\right)$ is the flow map and $\mathrm{J}(x,m)=\left(x,M(x,m)\right)$ is the jump map.
\section{main result}
In this section, we establish the safety and stability properties of our hybrid closed-loop system \eqref{Cl_ctrl} as well as the optimality of the path and the continuity of the control. 
The next lemma guarantees that our system is well-posed.
\begin{lemma}\label{lem2}
   The hybrid closed-loop system \eqref{Cl_ctrl} with the data $\mathcal{H}$, satisfies the basic conditions in \cite[Assumption 6.5]{Sanfelice}.
\end{lemma}
\begin{proof}
See Appendix \ref{appendix:Lem2}.
\end{proof}
Now, let us define the augmented state space and the desired equilibrium set as follows: 
\begin{align}\label{Hyb_freespace_att}
    \mathcal{K}:=\mathcal{X}\times\{-1,0,1\},\;
    \mathcal{A}:=\{x_d\}\times\{0\}.
\end{align}
For the hybrid closed-loop system \eqref{Cl_ctrl}, the augmented state space $\mathcal{K}$ is forward invariant as formally stated in the following lemma.
\begin{lemma}\label{lem3}
    The augmented state space $\mathcal{K}$, defined in \eqref{Hyb_freespace_att}, is forward invariant for the hybrid closed-loop system \eqref{Cl_ctrl}. 
\end{lemma}
\begin{proof}
See Appendix \ref{appendix:Lem3}.
\end{proof}
Finally, we can present our main result in the next theorem.
\begin{theorem}\label{the1}
    Consider the augmented state space $\mathcal{K}$, described in \eqref{Hyb_freespace_att}, and the hybrid closed-loop system \eqref{Cl_ctrl}. Then, the following statements hold:
    \begin{itemize}
        \item [i)] The augmented state space $\mathcal{K}$ is forward invariant.
        \item [ii)] The set $\mathcal{A}$ is globally asymptotically stable.
    \end{itemize}
\end{theorem}
\begin{proof}
See Appendix \ref{appendix:the1}.
\end{proof}
Theorem \ref{the1} states that for any initial condition inside the free space $\mathcal{X}$, the point-mass robot with position $x$ reaches its destination $x_d$ safely. One important characteristic of our proposed hybrid controller is that the resultant motion, during the flow, is two-dimensional for all $n\geq2$ which is stated in the following lemma.
\begin{lemma}\label{lem4}
Consider the augmented state space $\mathcal{K}$ defined in \eqref{Hyb_freespace_att}. The trajectory $x(t,j)$ of the closed-loop system \eqref{Cl_ctrl}
belongs to the two-dimensional plane $\mathcal{PL}(x_d^m,c,x_0^m)$ passing by the points $x_d^m$, $c$, and $x_0^m=x(t_0^m,j_0^m)\in\mathcal{X}\setminus\mathcal{L}(x_d,c)$ for $m\in\{-1,0,1\}$, $0<t_0^m\leq t\leq t_f^m$, and $0<j_0^m\leq j\leq j_f^m$ where $(t_0^m, j_0^m)$ and $(t_f^m,j_f^m)$ are, respectively, the activation and deactivation hybrid time of mode $m$. If $x_0^m\in\mathcal{X}\cap\mathcal{L}(x_d,c)$, $x(t,j)\in\mathcal{PL}(x_d^m,c,y)$ for all $(t,j)\in[t_0^m,t_f^m]\times[j_0^m,j_f^m]$ and $m\in\{-1,0,1\}$, where $y\in\mathcal{X}\setminus\mathcal{L}(x_d,c)$.
\end{lemma}
\begin{proof}
See Appendix \ref{appendix:Lem4}.
\end{proof}
Lemma \ref{lem4} states that the robot motion during mode $m$ is planar which means that the motion can be in different planes with different modes. The next lemma provides the condition to have the full trajectory in a single plane.
\begin{lemma}\label{lem5}
    Consider the augmented state space $\mathcal{K}$ defined in \eqref{Hyb_freespace_att} and the hybrid closed-loop system \eqref{Cl_ctrl}. For every initial condition $(x(0,0),m(0,0))\in\mathcal{X}\setminus\mathcal{L}(x_d,c)\times\{-1,1\}$, if the virtual destinations are selected such that $x_d^m\in(\mathcal{C}^{=}(x_d,c-x_d,\theta(x_d))\setminus\mathcal{E}(x_d))\cap\mathcal{PL}(x_d,c,x(0,0))$, then $x(t,j)\in\mathcal{PL}(x_d,c,x(0,0))$ for all $t\geq0$ and $j\in\mathbb{N}$. If $(x(0,0),m(0,0))\in\mathcal{X}\cap\mathcal{L}(x_d,c)\times\{-1,1\}$, and the virtual destinations are selected as $x_d^m\in(\mathcal{C}^{=}(x_d,c-x_d,\theta(x_d))\setminus\mathcal{E}(x_d))\cap\mathcal{PL}(x_d,c,y)$, then $x(t,j)\in\mathcal{PL}(x_d,c,y)$ for all $t\geq0\,j\in\mathbb{N}$ where $y\in\mathcal{X}\setminus\mathcal{L}(x_d,c)$.
\end{lemma}
\begin{proof}
See Appendix \ref{appendix:Lem5}.
\end{proof}
The result of Lemma \ref{lem5} requires selecting virtual destinations for every initial condition. We should also mention that the initial conditions $(x(0,0),0)$ starting from the {\it straight} mode are omitted since the hybrid system, in this case, will have no jumps and will remain in the same mode, meaning that the entire trajectory will be on a single plane regardless of the virtual destinations.
\begin{rem}\label{rem1}
    The closed-loop hybrid system \eqref{Cl_ctrl} can achieve the same optimality as in \eqref{18} under the following conditions:
    \begin{itemize}
        \item [i)]When the initial position $x(0,0)$ is in the hysteresis region, initialize the mode $m$ as follows:
    \begin{itemize}
    \item If $x(0,0)\in\mathcal{F}_k\cap\mathcal{F}_0$, $k\in\{-1,1\}$, $m=0$.
    \item If $x(0,0)\in\left(\bigcap\limits_{k\in\{-1,1\}}\mathcal{F}_k\setminus\mathcal{F}_0\right)\cap\mathcal{P}_{<}(c,w)$, $m=1$.
    \item If $x(0,0)\in\left(\bigcap\limits_{k\in\{-1,1\}}\mathcal{F}_k\setminus\mathcal{F}_0\right)\cap\mathcal{P}_{>}(c,w)$, $m=-1$.
    \item If $x(0,0)\in\left(\bigcap\limits_{k\in\{-1,1\}}\mathcal{F}_k\setminus\mathcal{F}_0\right)\cap\mathcal{P}_{=}(c,w)$, $m\in\{-1,1\}$.
    \end{itemize}
    where $w=x_d^{-1}-x_d^1$.
    \item [ii)]Thanks to the placement of the virtual destinations on the hat of the cone enclosing the obstacle and the scalar function $\mu(x,m)$, the control input enjoys continuity in 2D which will be shown in Lemma \ref{lem6}. If in addition, the strategy in item i) is implemented, the generated path is the shortest.
    \item [iii)]To enjoy both continuity and optimality in 3D, we select the virtual destinations as in Lemma \ref{lem5} for every initial condition, which brings us to the 2D case, and implement the scheme in item i). 
    \end{itemize}
\end{rem}
The implementation of the hybrid control \eqref{hyb_ctrl} following the instructions of Remark \ref{rem1} is summarized in Algorithm \ref{alg1}.
\begin{algorithm}
 \caption{Implementation of the hybrid control \eqref{hyb_ctrl}}\label{alg1}
 \begin{algorithmic}[1]
 \renewcommand{\algorithmicrequire}{\textbf{Initialization:}}
 \REQUIRE $x_d,\,c,\,r,\,e,\,x(0,0),\,m(0,0)$. 
 \STATE Select $x_d^1$ and $x_d^{-1}$ as in Lemma \ref{lem5}. 
 \STATE Construct the flow and jump sets \eqref{flow_jump_sets}.
\STATE Measure $x$. 
  \STATE Update $m$ using \eqref{jump_map}.
  \STATE Execute $u(x,m)$ using \eqref{hyb_ctrl}.
  \STATE Return to step 3.
 \end{algorithmic} 
 \end{algorithm}
 In the next lemma, we state the continuity of the control input \eqref{hyb_ctrl} in two-dimensional spaces.
 \begin{lemma}\label{lem6}
     Consider the augmented state space $\mathcal{K}$ of dimension $2\times1$ defined in \eqref{Hyb_freespace_att}. The hybrid control law \eqref{hyb_ctrl}, when the mode $m$ is initialized as in Remark \ref{rem1}, is continuous.
 \end{lemma}
 \begin{proof}
See Appendix \ref{appendix:Lem6}.
\end{proof}
\section{Numerical simulation}
In this section, we demonstrate the effectiveness of our approach by means of numerical simulations carried out in two spaces. The first space is two-dimensional with an obstacle centered on $c_1=[0\,-5]^{\top}$ and has a radius $r_1=2$. The second space is three-dimensional with an obstacle centered at $c_2=[1\,1\,1]^{\top}$ and has a radius $r_2=0.7$. In this space, we compare the trajectories obtained by our control with those generated by the hybrid control proposed in \cite{HybBerkaneECC2019}. We use the same design parameters as in \cite{HybBerkaneECC2019}. In both scenarios, the origin is the destination (red marker), and $e=\|x_d-x_d^m\|=0.1$ for $m=\pm 1$. The results are shown in Fig. \ref{fig:sim2d}, \ref{fig:sim3d}, and a simulation video of a three-dimensional scenario illustrating the process of selecting appropriate virtual destinations with different initial positions can be found at \url{https://youtu.be/AoyMIYIL1Qs}. As observed from Figure~\ref{fig:sim3d}
the trajectories obtained using our algorithm follow the shortest path to the destination but still guarantee convergence to the target from all initial conditions as in \cite{HybBerkaneECC2019}.
\begin{figure}[!h]
\centering
\includegraphics[scale=0.25]{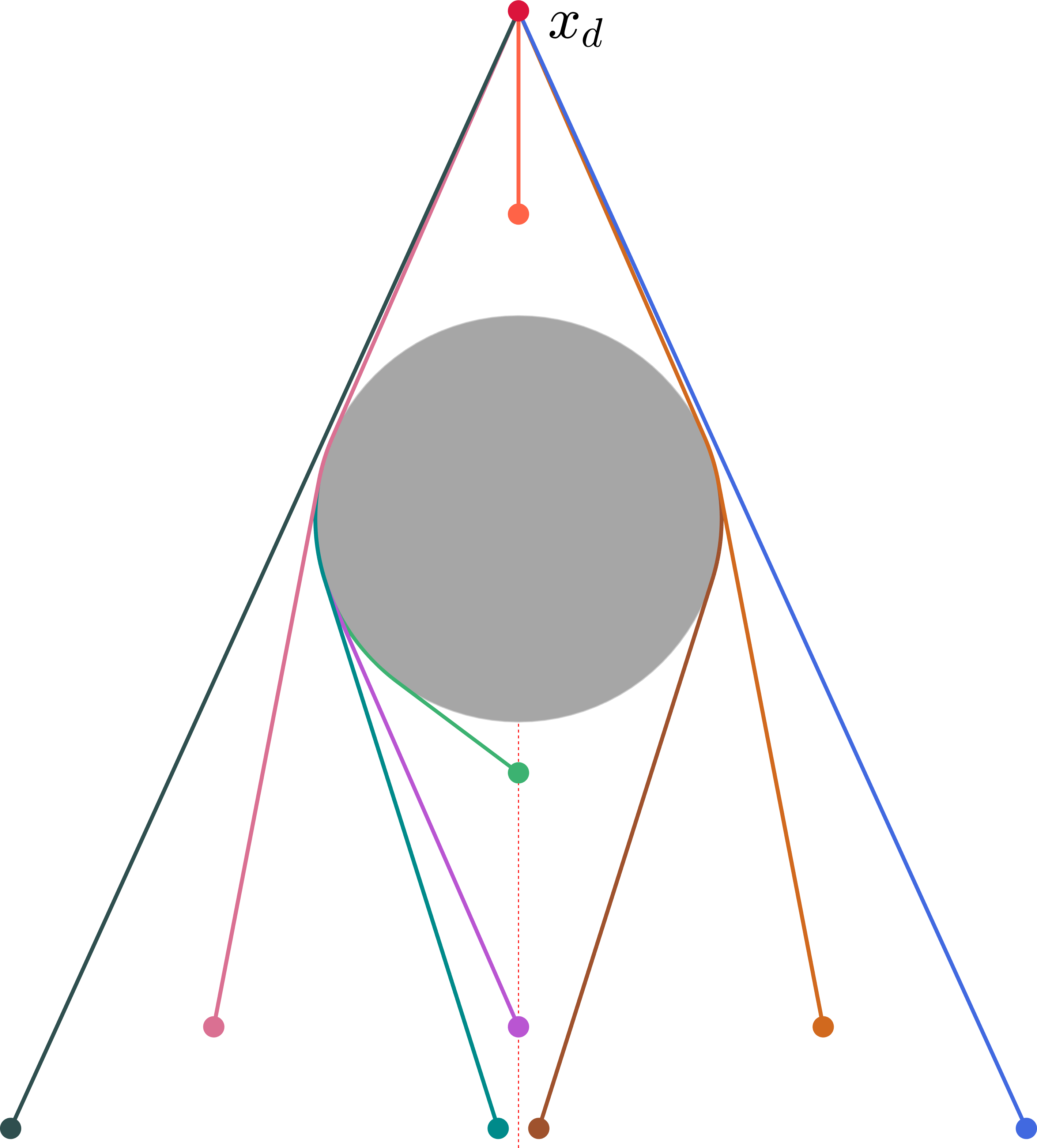}
\caption{Trajectories generated  by our proposed hybrid control law in 2D space from nine different initial conditions.}
\label{fig:sim2d}
\end{figure}
\begin{figure}[!h]
\centering
\includegraphics[scale=0.34]{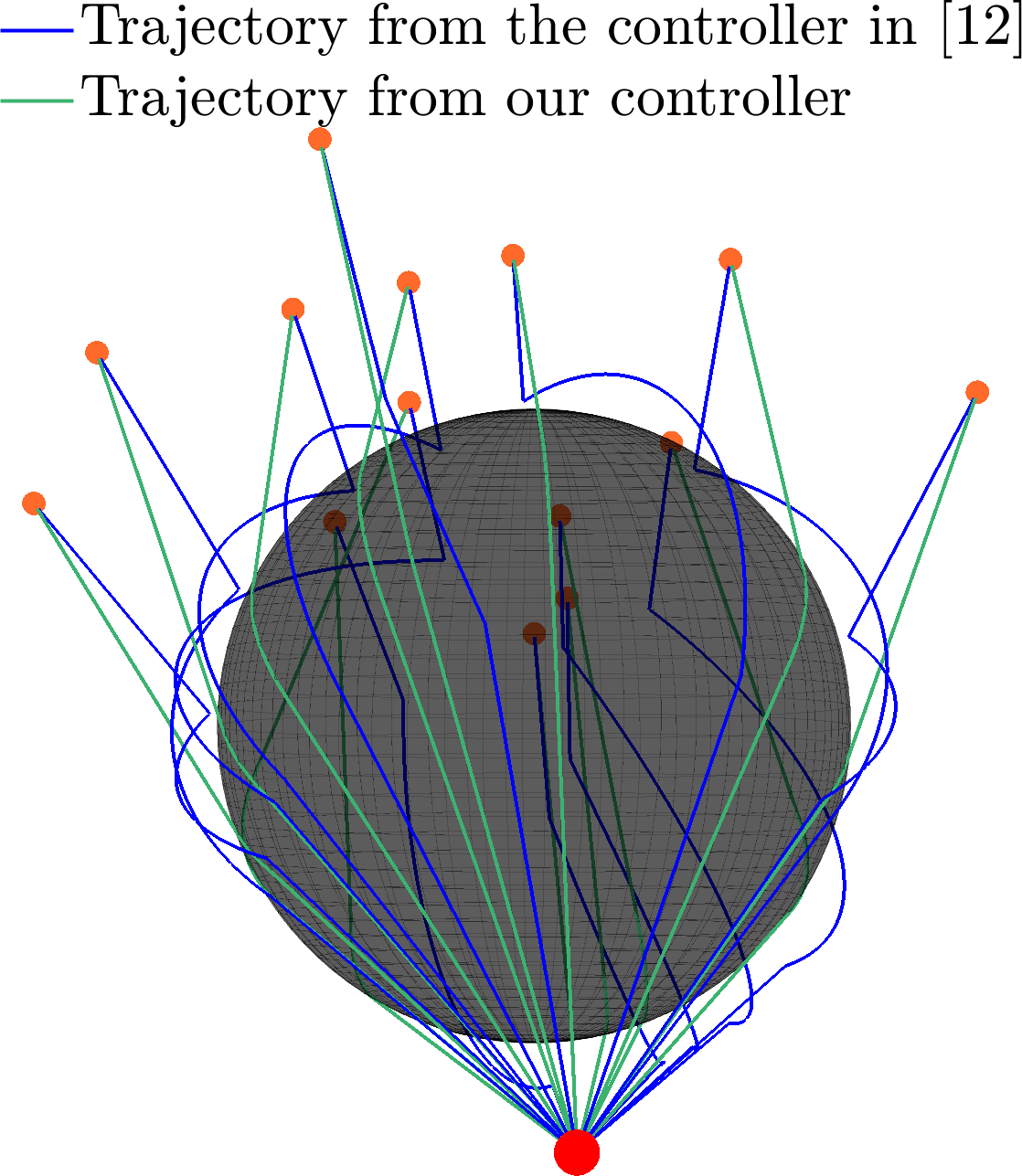}
\caption{Comparison of our proposed hybrid control law with the hybrid control proposed in \cite{HybBerkaneECC2019} in 3D space.}
\label{fig:sim3d}
\end{figure}
\section{Conclusion}
We have proposed a hybrid feedback control in $n$-dimensional spaces with a single spherical obstacle that guarantees global asymptotic stability of a target location and the optimality of the generated path. Two virtual destinations are considered to avoid the undesired equilibria issue via an appropriately designed switching mechanism that allows the robot to perform the obstacle avoidance maneuver along the shortest path within the cone enclosing the obstacle. Once the robot has a clear line of sight to the target, it heads straight to it. Moreover, the proposed control preserves continuity when switching between different operation modes. Taking into account multiple obstacles with complex shapes would be an interesting extension of the present work.  
\section{APPENDIX}
\subsection{Proof of Lemma \ref{lem2}}\label{appendix:Lem2}
The sets $\mathcal{F}$ and $\mathcal{J}$ defined in \eqref{flow_jump_sets} are closed subsets of $\mathbb{R}^n\times\{-1,0,1\}$. $\mathrm{F}$ is continuous on $\mathcal{F}$. To show that $\mathrm{J}$ is outer semi-continuous and locally bounded relative to $\mathcal{J}$, we first need to prove that $\mathrm{J}(x,m)\neq\varnothing$ for all $(x,m)\in\mathcal{J}$, which holds if $M(x,m)\neq\varnothing$ for all $(x,m)\in\mathcal{J}$. Given \eqref{jump_map}, we only need to show $M(x,0)\neq\varnothing$ for $(x,m)\in\mathcal{J}_0\times\{0\}$ which is true if $B(x)=\varnothing$ for all $x\in\mathbb{R}^n$. Since the angles $\varphi_1,\varphi_{-1}$ are selected according to Lemma \ref{lem1}, we have $\mathcal{C}_{\mathbb{R}^n}^{\leq}(c,v_1,\varphi_1)\cap\mathcal{C}_{\mathbb{R}^n}^{\leq}(c,v_{-1},\varphi_{-1})=\{c\}$, which implies that $\mathcal{C}_{\mathbb{R}^n}^{\geq}(c,v_1,\varphi_1)\cap\mathcal{C}_{\mathbb{R}^n}^{\geq}(c,v_{-1},\varphi_{-1})=\mathbb{R}^n$. Therefore, in view of \eqref{jump_map_B}, $B(x)\neq\varnothing$ for all $x\in\mathbb{R}^n$, and hence $\mathrm{J}(x,m)\neq\varnothing$ for all $(x,m)\in\mathcal{J}$. Furthermore, $\mathrm{J}$ has a closed graph as $B$ is allowed to be set-valued whenever $x\in\cap_{i=-1,1}\mathcal{C}_{\mathbb{R}^n}^{\leq}(c_k,v^k_i,\varphi_m)$. Thus,
according to \cite[Lemma 5.10]{Sanfelice}, $\mathrm{J}$ is outer semi-continuous and locally bounded relative to $\mathcal{J}$.
\subsection{Proof of Lemma \ref{lem3}}\label{appendix:Lem3}
First of all, it is clear from the definition of the flow and jump sets \eqref{fig:set_def} that their union covers the augmented state space, \textit{i.e.}, $\mathcal{K}=\mathcal{F}\cup\mathcal{J}$. Now, define the set $\mathcal{S}_{\mathcal{H}}(\mathcal{K})$ of all maximal solutions to the hybrid system \eqref{Cl_ctrl} represented by its data $\mathcal{H}$ with $\phi(0,0)\in\mathcal{K}$. Each solution $\phi \in\mathcal{S}_{\mathcal{H}}(\mathcal{K})$ has range $\text{rge}\,\phi \subset\mathcal{K}$. The augmented state space $\mathcal{K}$ is forward invariant for $\mathcal{H}$ if, for each $\phi(0,0) \in \mathcal{K}$, there exists one solution, and every $\phi\in \mathcal{S}_{\mathcal{H}}(\mathcal{K})$ is complete and has range $\text{rge}\,\phi \subset\mathcal{K}$ as per \cite[Definition~3.3]{Forward_invariance_sanfelice}. The forward invariance of $\mathcal{K}$ is then shown by the completeness of the solutions $\phi\in\mathcal{S}_{\mathcal{H}}(\mathcal{K})$ which we prove using \cite[Proposition~6.10]{Sanfelice}. We start by showing the following viability condition
\begin{align}\label{viability_cdt}
    \mathrm{F}(x,m)\cap\mathrm{T}_{\mathcal{F}}(x,m), \forall (x,m)\in\mathcal{F}\setminus\mathcal{J},
\end{align}
where $\mathrm{T}_{\mathcal{F}}(x,m)$ is Bouligand's tangent cone of the set $\mathcal{F}$ at position $x$ as defined in \cite[Definition~5.12]{Sanfelice}. Inspired by \cite[Appendix~1]{SoulaimaneHybTr}, we proceed as follows. Let $(x,m)\in\mathcal{F}\setminus\mathcal{J}$, which implies by \eqref{flow_jump_sets} that $x\in\mathcal{F}_m\setminus\mathcal
{J}_m$ for some $m\in\{-1,0,1\}$. Consider the two cases (modes) $m=0$ and $m=\{-1,1\}$. In the first mode $m=0$, when $x\in\mathring{\mathcal{F}_0}\setminus\mathcal{J}_0$, $\mathrm{T}_{\mathcal{F}}(x,0)=\mathbb{R}^n\times\{0\}$ and \eqref{viability_cdt} holds. As $\partial\mathcal{F}_0\setminus\mathcal{J}_0=\partial\mathcal{O}\setminus\mathcal{J}_0$, according to \eqref{18} and \eqref{set1}, $\mathrm{T}_{\mathcal{F}}(x,0)=\mathcal{P}_{\geq}(x,x-c)$ and $(c-x)^{\top}(x_d-x)<0$ where $\mathcal{P}_{=}(x,(x-c))$ is the hyperplane tangent to the obstacle at position $x$. Since $u(x,0)=-\gamma(x-x_d)$, $u(x,0)^{\top}(x-c)>0$ and \eqref{viability_cdt} holds. In the second mode $m\in\{-1,1\}$, for $x\in\mathring{\mathcal{F}_m}\setminus\mathcal{J}_m$, $\mathrm{T}_{\mathcal{F}}(x,0)=\mathbb{R}^n\times\{0\}$ and \eqref{viability_cdt} holds. Since $\partial\mathcal{F}_m\setminus\mathcal{J}_m\subset(\partial\mathcal{O}\setminus\mathcal{J}_m)$, according to \eqref{18}, \eqref{set2}, and \eqref{set3}, $\mathrm{T}_{\mathcal{F}}(x,m)=\mathcal{P}_{\geq}(x,x-c)$ and $(c-x)^{\top}(x_d^m-x)\geq0$ where $\mathcal{P}_{=}(x,(x-c))$ is the hyperplane tangent to the obstacle. From \eqref{hyb_ctrl}, for $x\in\partial\mathcal{O}$, $\theta(x)=\frac{\pi}{2}$ and $u(x,m)=\gamma\mu(x,m)\left(x_d^m-x-\|x_d^m-x\|\cos(\beta_m(x))\frac{c-x}{\|c-x\|}\right)$ where $\beta_m(x)=\angle(c-x,x_d^m-x)$ and $\mu(x,m)>0$. Therefore, $u(x,m)^{\top}(x-c)=\gamma\mu(x,m)((x_d^m-x)^{\top}(x-c)+(x_d^m-x)^{\top}(x-c))=0$ and \eqref{viability_cdt} holds for $m\in\{-1,1\}$. Therefore, according to \cite[Proposition 6.10]{Sanfelice}, since \eqref{viability_cdt} holds
for all $(x,m)\in\mathcal{F}\setminus\mathcal{J}$, there exists a nontrivial solution to the hybrid system $\mathcal{H}$ for
every initial condition in $\mathcal{K}$. Now, we will show that the hybrid closed-loop system \eqref{Cl_ctrl} has no finite escape-time over the flows. Consider the Lyapunov function $V(x)=\frac{1}{2}||x-x_d^m||^2$ for $m\in\{-1,0,1\}$ whose time-derivative is given by
\begin{align*}
    \Dot{V}(x)&=\frac{\partial V(x)}{\partial x}^{\top}\Dot{x}\\
    &=\begin{cases}
    (x-x_d^m)^{\top}u_d^m(x),\qquad(x,m)\in\mathcal{F}_0\times\{0\}\\
 \mu(x,m)(x-x_d^m)^\top u_m(x),\\
 \qquad(x,m)\in\mathcal{F}_m\times\{-1,1\}
\end{cases}\\
&=\begin{cases}
    -\gamma (x-x_d^m)^{\top}(x-x_d^m),\qquad(x,m)\in\mathcal{F}_0\times\{0\}\\
    K(x-x_d^m)^{\top}\left(x_d^m-x-\tau_m(x)\frac{c-x}{\|c-x\|}\right),\\
 \qquad(x,m)\in\mathcal{F}_m\times\{-1,1\}
\end{cases}\\
&=\begin{cases}
    -\gamma \|x-x_d^m\|^2,\qquad(x,m)\in\mathcal{F}_0\times\{0\}\\
    -K\left(\|x-x_d^m\|^{2}-\tau_m(x)\frac{(x_d^m-x)^{\top}(c-x)}{\|c-x\|}\right),\\
 \qquad(x,m)\in\mathcal{F}_m\times\{-1,1\}
\end{cases}\\
&=\begin{cases}
    -\gamma \|x-x_d^m\|^2,\qquad(x,m)\in\mathcal{F}_0\times\{0\}\\
    -K\left(\|x-x_d^m\|^{2}-\|x_d^m-x\|\tau_m(x)\cos(\beta_m(x))\right),\\
 \qquad(x,m)\in\mathcal{F}_m\times\{-1,1\}
\end{cases}\\
&=\begin{cases}
    -\gamma \|x-x_d^m\|^2,\qquad(x,m)\in\mathcal{F}_0\times\{0\}\\
    -K\|x-x_d^m\|^{2}\frac{\sin(\beta_m(x))}{\sin(\theta(x))}\cos(\theta(x)-\beta_m(x)),\\
\qquad(x,m)\in\mathcal{F}_m\times\{-1,1\}
\end{cases}
\end{align*}
where $K=\gamma\mu(x,m)$, and we used the fact that $\sin(\theta(x))-\sin(\theta(x)-\beta_m(x))\cos(\beta_m(x))=\sin(\beta_m(x))\cos(\theta(x)-\beta_m(x))$, $0<\theta(x)\leq\frac{\pi}{2}$ and $0\leq\beta_m(x)\leq\theta(x)$. It is clear that $\Dot{V}(x)\leq0$ for all $(x,m)\in\mathcal{F}_m\times\{-1,0,1\}$.
Hence, finite escape time can not occur for $x\in\mathcal{F}_m$ as this would make $(x-x_d^m)^{\top}(x-x_d^m)$ grow unbounded and would contradict the fact that $\Dot{V}(x)\leq 0$. Therefore, all maximal solutions do not have finite escape times. Moreover, according to \eqref{jump_dyn},
$x^+ = x$, and from the definitions \eqref{jump_map},
\eqref{jump_map_B}, it follows that $\mathrm{J}(\mathcal{J})\subset \mathcal{K}$. Thus, solutions of the hybrid system \eqref{Cl_ctrl} cannot leave $\mathcal{K}$ through jumps and, as per \cite[Proposition 6.10]{Sanfelice}, all maximal solutions
are complete.
\subsection{Proof of Theorem \ref{the1}}\label{appendix:the1}
Item i) follows directly from Lemma \ref{lem3}. Using \cite[Definition~7.1]{Sanfelice}, we prove item ii) by first showing the stability of $\mathcal{A}$, and then its global attractivity. Since $x_d\in\mathring{\mathcal{X}}$, there exists $\bar{\epsilon}>0$ such that $\mathcal{B}(x_d,\bar{\epsilon})\cap\mathcal{O}=0$. As per the sets definitions in \eqref{set1}, $\mathcal{B}(x_d,\epsilon)\subset\mathcal{F}_0$ for all $\epsilon\in[0,\bar{\epsilon}]$. Thus, $\mathcal{B}(x_d,\epsilon)\cap\mathcal{J}_0=\varnothing$, and $x$ evolves under $\Dot{x}=-\gamma(x-x_d)$, which implies forward invariance of the set $\mathcal{B}:=\mathcal{B}(x_d,\epsilon)\times\{-1,0,1\}$. Therefore, according to \cite[Definition~7.1]{Sanfelice}, the set $\mathcal{A}$ is stable for the hybrid system \eqref{Cl_ctrl}. Now, let us prove the global attractivity of $\mathcal{A}$ by showing that for all initial condition $(x(0,0),m(0,0))\in\mathcal{F}_m\times\{-1,1\}$, there exists a hybrid time $t>0$, $j\in\mathbb{N}\setminus\{0\}$ such that $x(t,j)\in\mathcal{F}_0$. Consider the positive definite function $V(x)=\frac{1}{2}||x-x_d^m||^2$ defined in Appendix \ref{appendix:Lem3} for $m\in\{-1,1\}$. Given that $\Dot{V}(x)\leq0$, as shown in Appendix \ref{appendix:Lem3}, and $\Dot{V}(x)=0$ only if $x\in\mathcal{L}(c,x_d^m)$ ($\beta_m(x)=0$), which is excluded from the set $\mathcal{F}_m$ for $m\in\{-1,1\}$, one can conclude that $V(x)<0$ for all $(x,m)\in\mathcal{F}_m\times\{-1,1\}$. Therefore, as $x_d^m\in\mathcal{F}\setminus\mathcal{F}_m$, there exists a hybrid time $t>0$, $j\in\mathbb{N}\setminus\{0\}$,  such that $x(t,j)$ will leave $\mathcal{F}_m$ and the mode $m(t,j)$ will switch to the {\it straight} mode where the flow $\Dot{x}=-\gamma(x-x_d)$ guarantees the global attractivity of the equilibrium set $\mathcal{A}$.
\subsection{Proof of Lemma \ref{lem4}}\label{appendix:Lem4}
\textbf{Case 1}: The {\it straight} mode ($m=0$): the velocity vector is given by $u(x,0)=\gamma(x_d^0-x)$, which implies that the trajectory is the line segment joining $x(t_0^0,j_0^0)$ and $x(t_f^0,j_f^0)=x_d^0=x_d$, $t_f^0,j_f^0\to\infty$. Therefore, if $x_0^0\in\mathcal{X}\setminus\mathcal{L}(x_d,c)$ \big(resp. $x_0^0\in\mathcal{X}\cap\mathcal{L}(x_d,c)$\big), $x(t,j)\in\mathcal{PL}(x_d^0,c,x_0^0)$ \big(resp. $x(t,j)\in\mathcal{PL}(x_d^0,c,y)$, $y\in\mathcal{X}\setminus\mathcal{L}(x_d,c)$\big) for all $(t,j)\in[t_0^0,t_f^0]\times[j_0^0,j_f^0]$. 
\\ \textbf{Case 2}: The {\it projection} mode ($m\in\{-1,1\}$): the velocity vector is given by $u(x,m)=\gamma(x_d^m-x)-\tau_m(x)\frac{c-x}{\|c-x\|}$, a function of the two vectors $(c-x)$ and $(x_d^m-x)$, implying that the motion is bi-dimensional. Therefore, if $x_0^0\in\mathcal{X}\setminus\mathcal{L}(x_d,c)$ \big(resp. $x_0^0\in\mathcal{X}\cap\mathcal{L}(x_d,c)$\big), the trajectory $x(t,j)\in\mathcal{PL}(x_d^m,c,x_0^m)$ \big(resp. $x(t,j)\in\mathcal{PL}(x_d^0,c,y)$, $y\in\mathcal{X}\setminus\mathcal{L}(x_d,c)$\big) for all $(t,j)\in[t_0^m,t_f^m]\times[j_0^m,j_f^m]$.
\subsection{Proof of Lemma \ref{lem5}}\label{appendix:Lem5}
Let $\bar{p}\in\mathbb{N}$, be the number of jumps of the hybrid system for an initial condition $(x(0,0),m)\in\mathcal{X}\times\{-1,1\}$. Let $\mathbb{T}_p:=[t_0^p,t_f^p]\times[j_0^p,j_f^p]$, $p\in\{0,\dots,\bar{p}\}$ be the hybrid time-interval between jump $p$ and jump $p+1$, and $\mathbb{T}_{\bar{p}}:=[t_0^{\bar{p}},+\infty]\times[j_0^{\bar{p}},+\infty]$ be the last hybrid time-interval. Note that $x_d^0=x_d$ and $x_0^p=x(t_0^p,j_0^p)$. We prove by induction the following two cases:\\
\textbf{Case 1}: $(x(0,0),m(0,0))\in\mathcal{X}\setminus\mathcal{L}(x_d,c)\times\{-1,1\}$. Let the virtual destinations be selected as $x_d^m\in\left(\mathcal{C}^{=}(x_d,c-x_d,\theta(x_d))\setminus\mathcal{E}(x_d)\right)\cap\mathcal{PL}(x_d,c,x(0,0))$. Define property $P(p)$ as follow: the trajectory after jump $p$ is included in the plane $\mathcal{PL}(x_d,c,x(0,0))$ during the hybrid time interval $\mathbb{T}_p$. Let us show that $P(0)$ is true. As per Lemma \ref{lem4}, $x(t)\in\mathcal{PL}(x_d^m,c,x_0^0)$ for $t\in\mathbb{T}_0$. Since $x_0^0=x(0,0)$ and $x_d^m,x(0,0)\in\mathcal{PL}(x_d,c,x(0,0))$, $\mathcal{PL}(x_d^m,c,x_0^0)=\mathcal{PL}(x_d,c,x(0,0))$. Assuming that $P(p)$ is true,  we will show that $P(p+1)$ is also true. Since $x(t)\in\mathcal{PL}(x_d^m,c,x_0^{p+1})$ for $t\in\mathbb{T}_{p+1}$ and  $m\in\{-1,0,1\}$, as per Lemma \ref{lem4}, and $x_0^{p+1}=x(t_f^{p})$, $x_d^m,x(t_f^{p})\in\mathcal{PL}(x_d,c,x(0,0))$, one can conclude that $x(t,j)\in\mathcal{PL}(x_d,c,x(0,0))$ for $t\in\mathbb{T}_{p+1}$. Therefore, $P(p+1)$ is true, completing the induction.\\
\textbf{Case 2}: $(x(0,0),m(0,0))\in\mathcal{X}\cap\mathcal{L}(x_d,c)\times\{-1,1\}$. Let the virtual destinations be selected as $x_d^m\in(\mathcal{C}^{=}(x_d,c-x_d,\theta(x_d))\setminus\mathcal{E}(x_d))\cap\mathcal{PL}(x_d,c,y)$ where $y\in\mathcal{X}\setminus\mathcal{L}$. The same proof as in \textbf{case 1)} is reproduced with the plane of motion $\mathcal{PL}(x_d,c,y)$.\\
Since $\cup_{p=0}^{\bar{p}}\mathbb{T}_p=[0,+\infty)\times\mathbb{N}$,  $x(t,j)\in\mathcal{PL}(x_d,c,x(0,0))$ for $t\geq0,j\in\mathbb{N}$, in \textbf{case 1} (resp. $x(t,j)\in\mathcal{PL}(x_d,c,y)$ for $t\geq0$ and $y\in\mathcal{X}\setminus\mathcal{L}$ in \textbf{case 2}), which concludes the proof.
\subsection{Proof of Lemma \ref{lem6}}\label{appendix:Lem6}
As the control law $u(x,m)$ is continuous during the flow and only switches from the {\it projection} mode to the {\it straight} mode, we only need to verify its continuity at the switching instances, which corresponds to the \textit{exit set}. Since the virtual destinations are on the hat of the cone enclosing obstacle $\mathcal{O}$ and following the initialization scheme of item i) in Remark \ref{rem1}, the jump will occur when $x\in\mathcal{E}(x_d)\cap\mathcal{E}(x_d^m)$ for $m\in\{-1,1\}$. Thus, since the space is two-dimensional, $\frac{(x_d-x)^\top}{\|x_d-x\|}\frac{(x_d^m-x)}{\|x_d^m-x\|}=1$ and $\|x_d-x\|=\|x_d^m-x\|+\|x_d-x_d^m\|$ ({\it i.e.,} $x,x_d$, and $x_d^m$ are aligned). In addition, $\beta_m(x)=\angle(c-x,x_d^m-x)=\theta(x)$ when $x\in\mathcal{E}(x_d^m)$. Therefore, $u(x,m)=\mu(x,m)u_d^m(x)=\gamma\frac{\|x-x_d^m\|+e}{\|x-x_d^m\|}(x_d^m-x)$, and since $e=\|x_d-x_d^m\|$, one has $u(x,m)=\gamma\|x_d-x\|\frac{x_d^m-x}{\|x_d^m-x\|}=\gamma(x_d-x)=u_d(x)=u(x,0)$. 

\addtolength{\textheight}{-12cm}   

\bibliographystyle{ieeetr}
\bibliography{references}
\end{document}